\documentclass[11pt,letter]{article}
\usepackage[utf8]{inputenc}
\usepackage{amsmath,amssymb,amsthm, mathabx, mathtools}
\usepackage[affil-sl]{authblk}
\usepackage{times}
\usepackage[margin=1in]{geometry}
\usepackage{hyperref}
\usepackage{xfrac}

\usepackage[colorinlistoftodos]{todonotes}

\newcommand{\lsignature}{\mathbf{\lambda}}
\newcommand{\surf}{\mathbb{S}}
\newcommand{\R}{\mathbb{R}}
\newcommand{\Z}{\mathbb{Z}}
\newcommand{\conv}{\operatorname{conv}}
\newcommand{\onevec}{\mathbf{1}}
\newcommand{\zerovec}{\mathbf{0}}

\newcommand{\cc}{\mathcal{C}}
\newcommand{\barcs}{\ensuremath{b}}
\newcommand{\fav}{\ensuremath{\eta}}

\usepackage{xparse} 
\DeclarePairedDelimiterX{\Set}[1]\{\}{\setargs{#1}}
\NewDocumentCommand{\setargs}{>{\SplitArgument{1}{|}}m}{\setargsaux#1}
\NewDocumentCommand{\setargsaux}{mm}{\IfNoValueTF{#2}{#1}{\nonscript\,#1\nonscript\;\delimsize\vert\allowbreak \nonscript\:\mathopen{}#2\nonscript\,}}

\newtheorem{theorem}{Theorem}
\newtheorem{proposition}[theorem]{Proposition}
\newtheorem{corollary}[theorem]{Corollary}
\newtheorem{lemma}[theorem]{Lemma}

\theoremstyle{definition}

\newtheorem{prob}[theorem]{Problem}

\title{Minimum-cost integer circulations in given homology classes}
\date{}
\author[1]{Sarah Morell}
\author[2]{Ina Seidel}
\author[2]{Stefan Weltge}
\affil[1]{Technische Universität Berlin, Germany}
\affil[1]{{\small \texttt{morell@math.tu-berlin.de}}}
\affil[2]{Technical University of Munich, Germany}
\affil[2]{{\small \texttt{\{ina.seidel,weltge\}@tum.de}}}

\makeatletter
\let\@@pmod\pmod
\DeclareRobustCommand{\pmod}{\@ifstar\@pmods\@@pmod}
\def\@pmods#1{\mkern4mu({\operator@font mod}\mkern 6mu#1)}
\makeatother

\begin{document}

\maketitle

\begin{abstract}
	Let~$D$ be a directed graph cellularly embedded in a surface together with non-negative cost on its arcs. Given any integer circulation in~$D$, we study the problem of finding a minimum-cost non-negative integer circulation in~$D$ that is homologous over the integers to the given circulation. 
    A special case of this problem arises in recent work on the stable set problem for graphs with bounded odd cycle packing number, in which the surface is non-orientable (Conforti et al., SODA'20).

	For \emph{orientable} surfaces, polynomial-time algorithms have been obtained for different variants of this problem. We complement these results by showing that the convex hull of feasible solutions has a very simple polyhedral description.

    In contrast, only little seems to be known about the case of \emph{non-orientable} surfaces. 
    We show that the problem is strongly NP-hard for general non-orientable surfaces, and give the first polynomial-time algorithm for surfaces of fixed genus.
	For the latter, we provide a characterization of~$\Z$-homology that allows us to recast the problem as a special integer program, which can be efficiently solved using proximity results and dynamic programming.
\end{abstract}

\section{Introduction}

Finding optimal subgraphs of a surface-embedded graph that satisfy certain topological properties is a basic subject in topological graph theory and an important ingredient in many algorithms, see, e.g.,~\cite[\S 1]{Erickson2011}.
Motivated by recent work of Conforti, Fiorini, Joret, Huynh, and the third author~\cite{ConfortiFHJW2019,ConfortiFHW2019}, we study a variant of the minimum-cost circulation problem with such an additional topological constraint. In~\cite{ConfortiFHJW2019,ConfortiFHW2019}, it was crucially exploited that the stable set problem for graphs with bounded genus and bounded odd cycle packing number can be efficiently reduced to the below-stated Problem~\ref{probMain}. While the standard minimum-cost circulation problem is among the most-studied problems in combinatorial optimization, much less seems to be known about this version.
\begin{prob}
    \label{probMain}
    Given a directed graph~$D$ cellularly embedded in a surface together with non-negative cost $c$ on its arcs and any integer circulation~$y$ in~$D$, find a minimum-cost non-negative integer circulation in~$D$ that is~$\Z$-homologous to~$y$.
\end{prob}
Here, a circulation~$x$ is said to be \emph{$ \Z$-homologous to~$y$} if their difference~$x - y~$is a linear combination of facial circulations with integer coefficients, where a \emph{facial circulation} is a circulation that sends one unit along the boundary of a single face, see Figure~\ref{fig:torus}. If~$x - y$ is a linear combination of facial circulations with real coefficients, we say that~$x$ is \emph{$ \R$-homologous} to~$y$. As an example, if~$y$ is the all-zeros circulation, then~$y$ itself is clearly an optimal solution to Problem~\ref{probMain}. However, for general~$y$ the all-zeros circulation might not be feasible. In fact, if the surface is different from the sphere and the projective plane, there are actually infinitely many homology classes, and their characterization is a basic subject in algebraic topology. We will provide more formal definitions in Section~\ref{sec:SurfacesHomology}.

\begin{figure}
    \centering
    \begin{minipage}[c]{0.3\textwidth}
        \includegraphics[width=\textwidth]{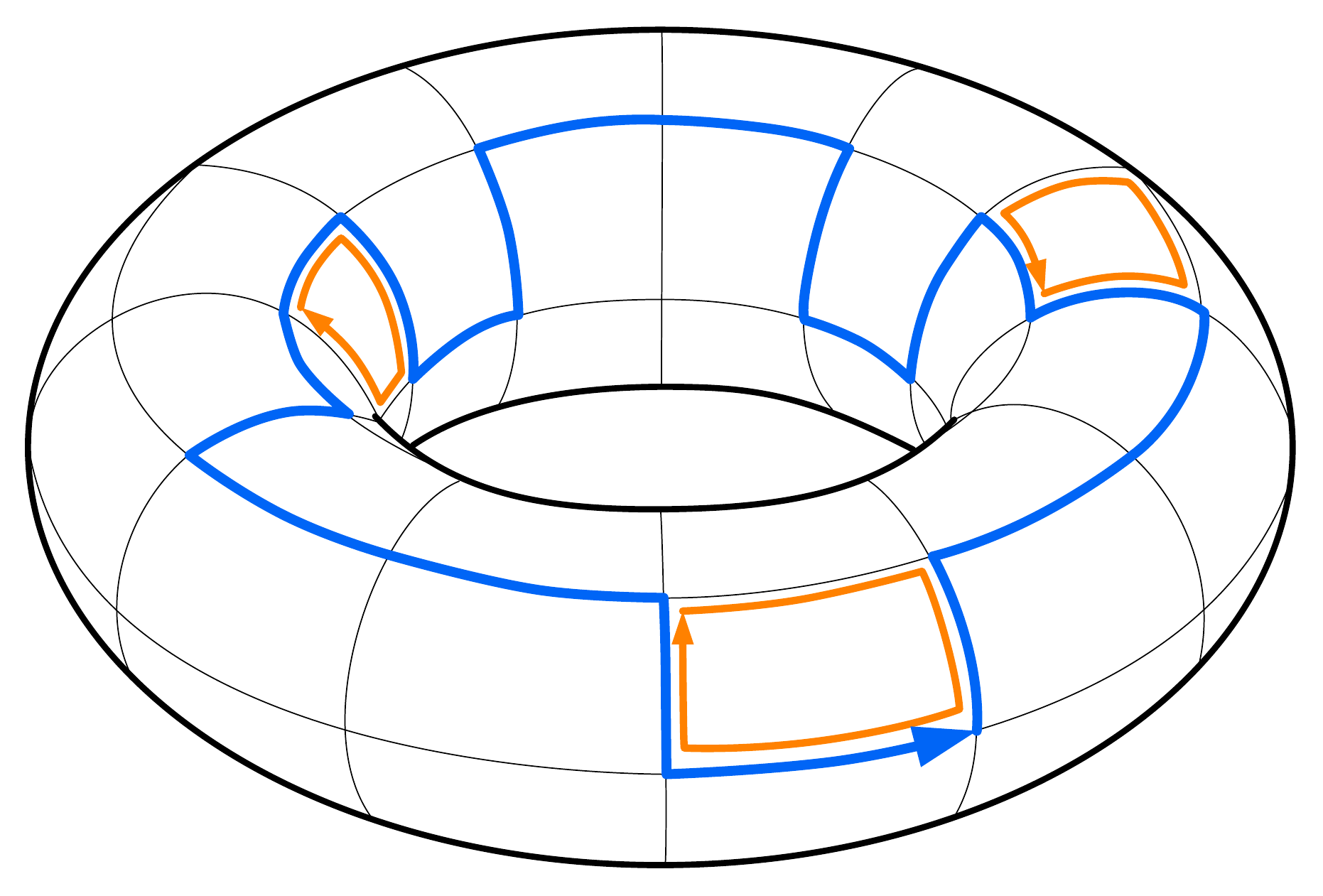}
    \end{minipage}
    \hspace{1cm}
    \begin{minipage}[c]{0.3\textwidth}
        \includegraphics[width=\textwidth]{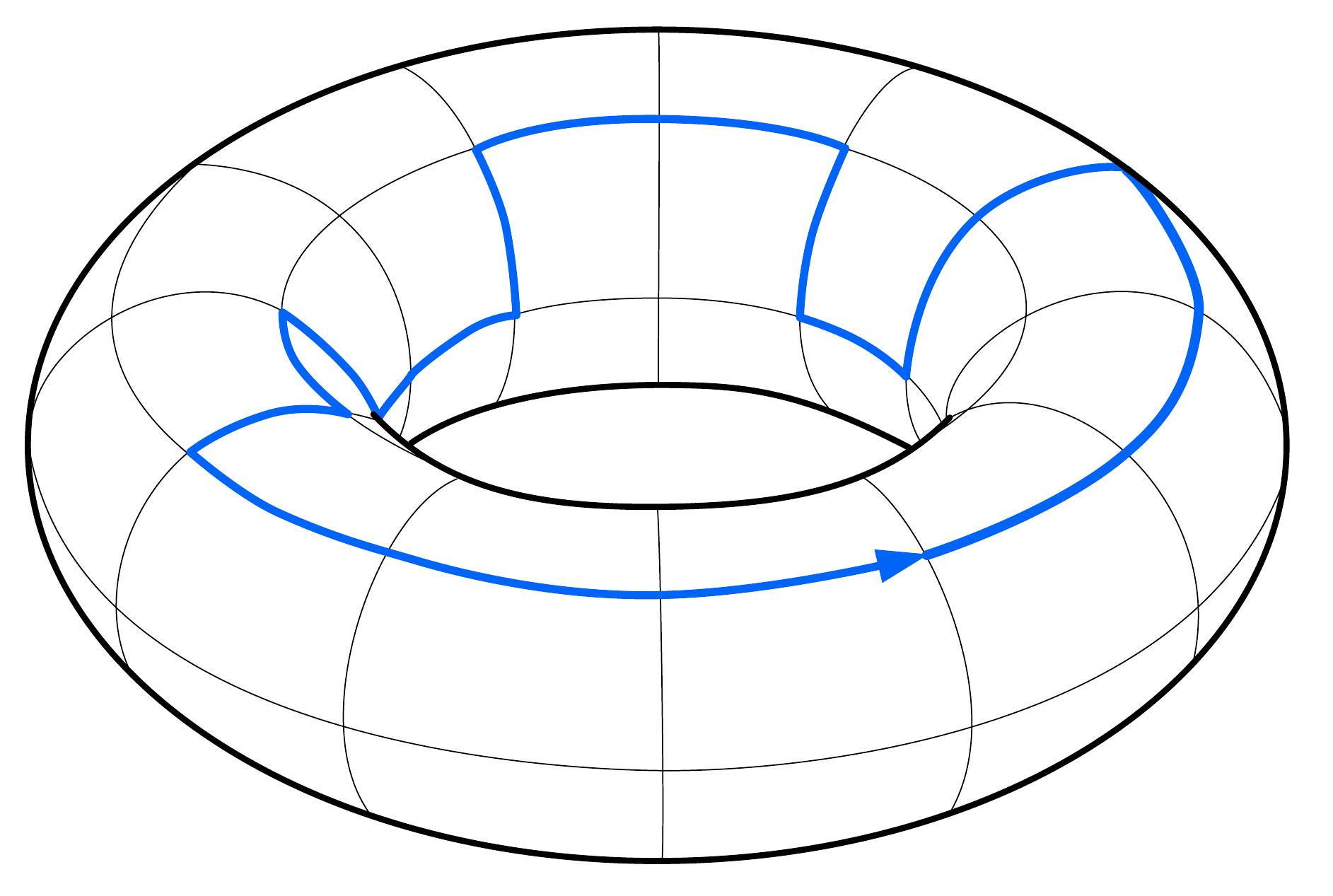}
    \end{minipage}
    \caption{The circulations that send one unit along the blue directed cycles on the torus are~$\Z$-homologous. In fact, their difference is the sum of three facial circulations which are depicted in orange.}
    \label{fig:torus}
\end{figure}

In this work, we introduce this problem to the combinatorial optimization community, with a particular emphasis on the case in which the surface is \emph{non-orientable}. While we complement existing results for orientable surfaces and show that the underlying polyhedra are actually easy to describe, only little seems to be known in the case of non-orientable surfaces. Our main result is a polynomial-time algorithm for non-orientable surfaces of fixed genus. Moreover, we show that the problem becomes NP-hard for general non-orientable surfaces.

For the case of \emph{orientable} surfaces, Chambers, Erickson, and Nayyeri~\cite{ChambersEN2012} show that Problem~\ref{probMain} can be solved in polynomial time. Their approach is based on an exponential-size linear program that can be solved using the ellipsoid method in near-linear time, provided that the surface has small genus. Dey, Hirani, and Krishnamoorthy~\cite{DeyHK2011} consider a variant of Problem~\ref{probMain} defined on simplicial complexes of arbitrary dimension, in which the (weighted)~$\ell_1$-norm of a chain homologous to~$y$ is to be minimized. For the case of an orientable surface, they derive a polynomial-time algorithm that is based on a linear program defined by a totally unimodular matrix. 

We complement these results by showing that the convex hull of feasible solutions to Problem~\ref{probMain} has a very simple polyhedral description. To this end, notice that Problem~\ref{probMain} asks for minimizing a linear objective over the convex hull of all non-negative integer circulations in~$D$ that are~$\Z$-homologous to~$y$. We will denote this polyhedron by~$P(D, y)$. Moreover, let~$P(D)$ be the convex hull of non-negative integer circulations in~$D$, which (as a network-flow polyhedron) has a simple linear description. Notice that any integer circulation~$x$ that is~$\Z$-homologous to~$y$ must also be~$\R$-homologous to~$y$. In other words,~$x$ must be contained in the affine subspace of all circulations that are~$\R$-homologous to~$y$, which we denote by~$L(D, y)$. Surprisingly, it turns out that it suffices to add the equations defining~$L(D, y)$ to a description of~$P(D)$ in order to obtain one for~$P(D, y)$.
\begin{theorem}
\label{thm:descriptionOfP(D,y)}
    Let~$D$ be a graph that is cellularly embedded in an orientable surface and let~$y$ be any integer circulation in~$D$. Then,~$P(D, y) = P(D) \cap L(D, y)$.
\end{theorem}
We will provide an explicit description for~$L(D, y)$ later. Unfortunately, Theorem~\ref{thm:descriptionOfP(D,y)} does not hold for non-orientable surfaces. In fact, we show that Problem~\ref{probMain} becomes inherently more difficult on general non-orientable surfaces.
\begin{theorem}
    \label{thm:hard}
    Problem~\ref{probMain} is strongly NP-hard on general non-orientable surfaces.
\end{theorem}
While the authors in~\cite{DunfieldH2011} show that variants of Problem~\ref{probMain} become NP-hard on~$3$-dimensional simplicial complexes, their approach does not seem to apply to surfaces. In fact, the reduction therein crucially relies on~$3$-dimensional gadgets, and equivalent~$2$-dimensional configurations are not obvious to us. 
We obtain a reduction from general~$3$-SAT instances, showing that Problem~\ref{probMain} is indeed (strongly) NP-hard.
To this end, we exploit ideas developed in~\cite{ConfortiFHJW2019} to reduce very particular instances of the stable set problem to Problem~\ref{probMain}.

On the positive side, we show that Problem~\ref{probMain} becomes tractable when dealing with non-orientable surfaces of \emph{fixed} genus:
\begin{theorem}
	\label{thm:existenceOfPolyAlgo}
	Problem~\ref{probMain} can be solved in polynomial time on non-orientable surfaces of fixed genus.
\end{theorem}
A special case of Problem~\ref{probMain} was already treated and shown to be solvable in polynomial time in~\cite{ConfortiFHJW2019}, where only instances arising from very specific stable set problems were considered. Here, we consider the general problem.
The algorithm in~\cite{ConfortiFHJW2019} is based on an alternative characterization of~$\Z$-homology, which we have to replace by a more general one, see Theorem~\ref{thm:zHomologousNonOrientable}.
In fact, in the instances considered in~\cite{ConfortiFHJW2019} the orientation of the arcs of~$D$ is already determined by an embedding scheme (which we define later) of the dual graph, which we cannot assume here.
Moreover, it is exploited that, in their setting, optimal circulations can be found in~$\{0,1\}^A$, which is also not the case for a general instance of Problem~\ref{probMain}.
Using a bound of Malni\v{c} and Mohar~\cite{MM92} on the number of certain non-freely-homotopic disjoint closed curves, it is then shown that an optimal circulation can be decomposed into few disjoint closed walks that can be enumerated efficiently.
In this approach, the existence of optimal solutions with small entries is crucial to obtain a polynomial running time.

We propose another decomposition technique that enables us to reformulate Problem~\ref{probMain} as an integer program in standard form with a constant number of equality constraints, provided that the genus is fixed.
Applying results on the proximity of integer programs, the resulting problem can be efficiently solved using dynamic programming.
Our approach does not rely on further topological ingredients, in particular we do not require bounds from~\cite{MM92}.

\paragraph{Outline}
In Section~\ref{sec:SurfacesHomology}, we provide a brief introduction to surfaces and graph embeddings, focusing on the facts that are necessary for the following two sections. Our polynomial-time algorithm for Problem~\ref{probMain} on non-orientable surfaces with fixed genus is described in Section~\ref{sec:fixedgenus_polytimealg}. The proof of Theorem~\ref{thm:hard} is presented in Section~\ref{sec:hardness}. In Section~\ref{secCharHomology}, we present alternative characterizations of homology. More precisely, Section~\ref{secOrientable} is devoted to the case of orientable surfaces and contains a discussion of Theorem~\ref{thm:descriptionOfP(D,y)}. The non-orientable case is treated in Section~\ref{secNonOrientable}. Here, we provide a proof for a main ingredient (Theorem~\ref{thm:zHomologousNonOrientable}) of our algorithm.
We close our paper with a discussion of open problems in Section~\ref{sec:openquestions}.

\section{Surfaces and Homology}
\label{sec:SurfacesHomology}

We start with a \emph{brief} introduction to surfaces, graph embeddings and the concept of homology. Further details and illustrations will be provided in Section~\ref{secCharHomology}.

A \emph{surface} is a non-empty connected compact Hausdorff topological space in which each point has an open neighborhood which is homeomorphic to the open unit disc in the plane. Examples of such surfaces are the sphere, the torus, and the projective plane. While the first two are \emph{orientable} surfaces, the latter one is \emph{non-orientable}. Up to homeomorphism, each surface~$\surf$ can be characterized by a single non-negative integer called the \emph{Euler genus}~$g$ of~$\surf$ together with the information whether $\surf$ is orientable. If~$\surf$ is orientable, then~$g$ is even and~$\surf$ can be obtained from the sphere by deleting~$\sfrac{g}{2}$ pairs of open discs and, for each pair, identifying their boundaries in opposite directions (``gluing handles''). Otherwise,~$\surf$ is non-orientable and can be obtained from the sphere by deleting~$g \ge 1$ open discs and, for each disc, identifying the antipodal points on its boundary (``gluing Möbius bands''), see Figure~\ref{fig:kleinBottle} for an illustration.

\begin{figure}
    \centering
    \begin{minipage}[c]{0.2\textwidth}
        \includegraphics[width=\textwidth]{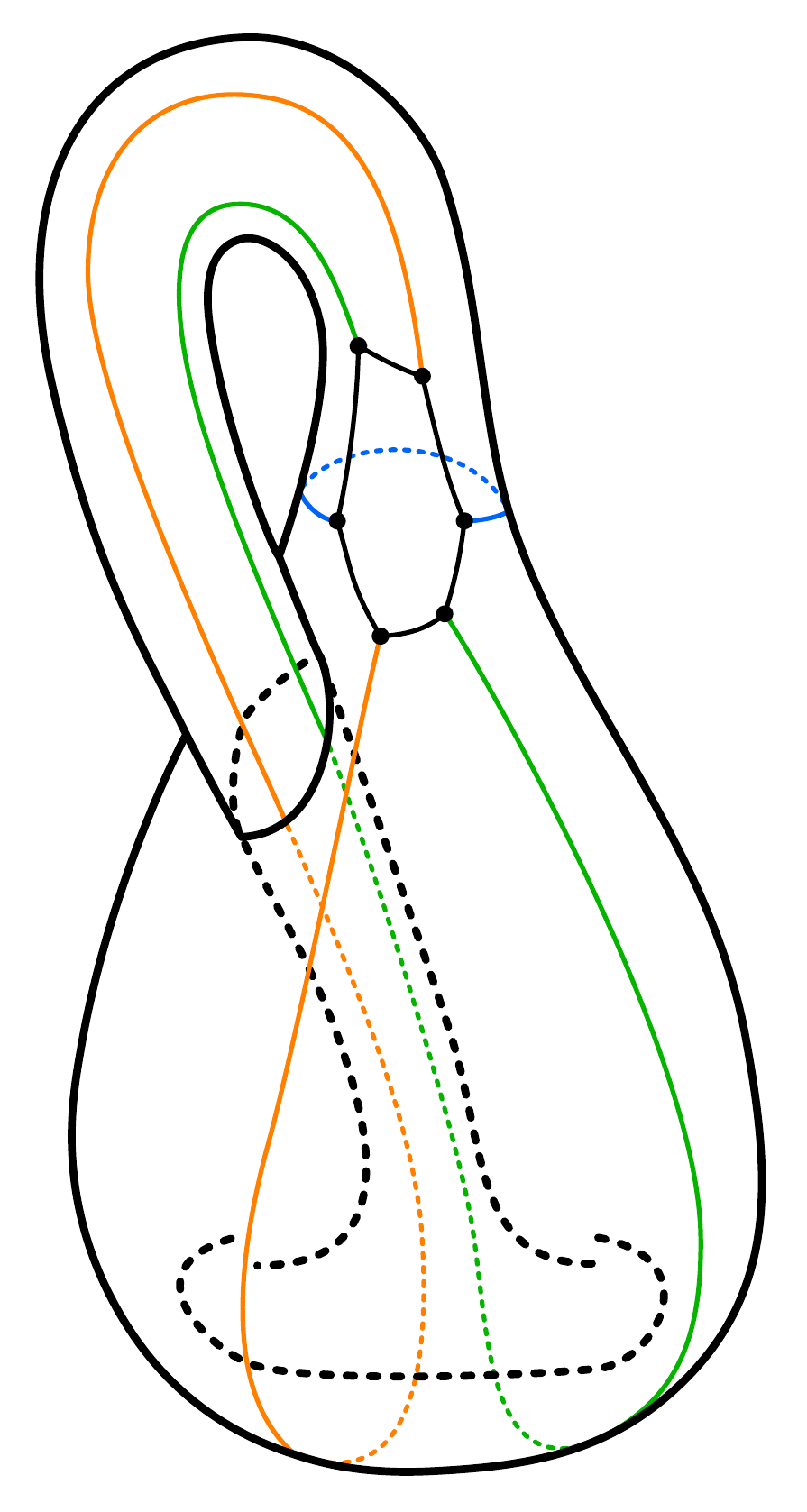}
    \end{minipage}
    \hspace{1cm}
    \begin{minipage}[c]{0.3\textwidth}
        \includegraphics[width=\textwidth]{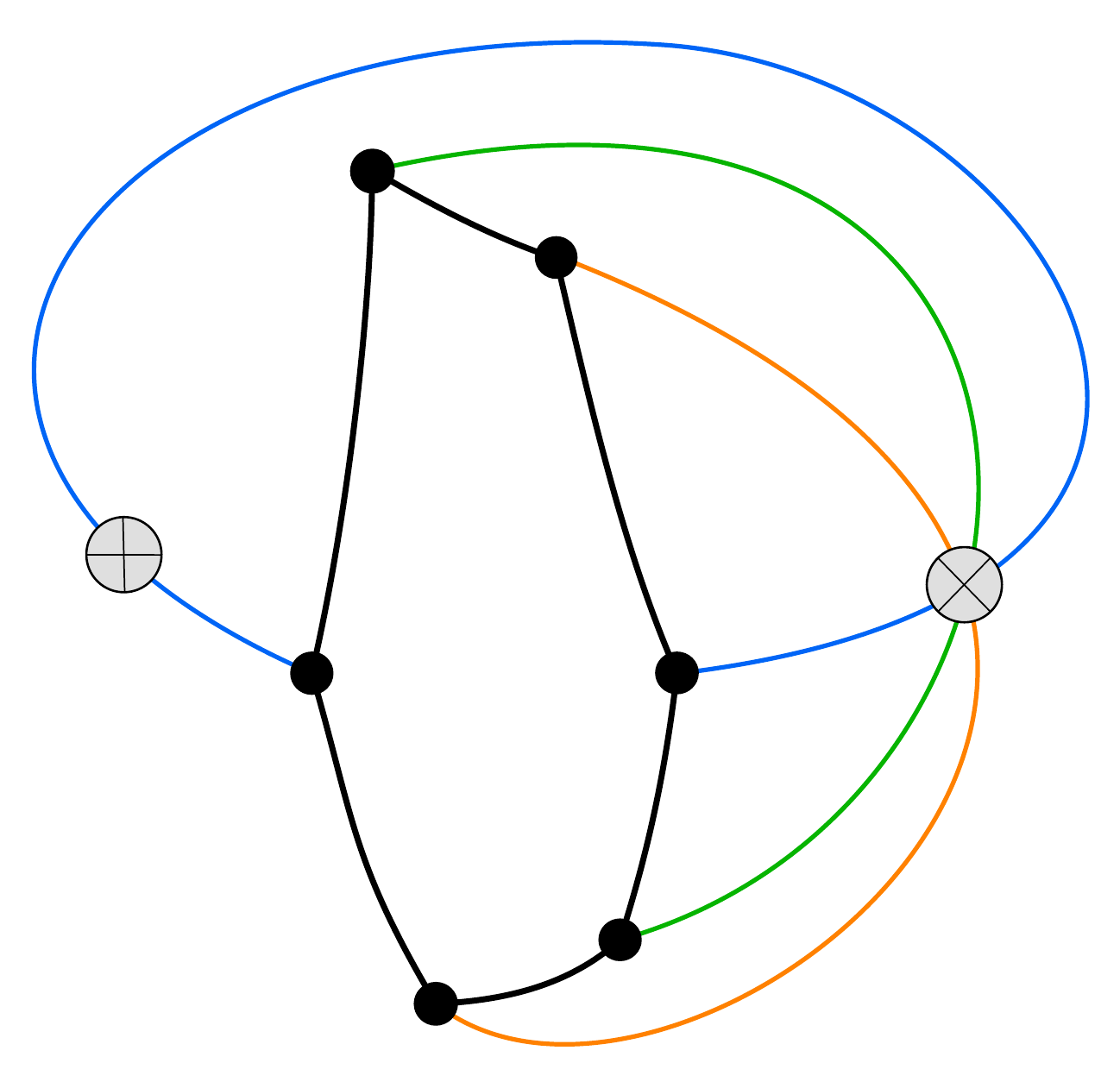}
    \end{minipage}
    \caption{A graph embedded in the \emph{Klein bottle}, the non-orientable surface of Euler genus~$2$. On the left, the surface is embedded in~$3$-dimensional space. Recall that the Klein bottle is obtained from the sphere by deleting two open discs and, for each disc, identifying the antipodal points on its boundary. On the right, an equivalent embedding of the same graph is shown, where these discs are depicted in gray.}
    \label{fig:kleinBottle}
\end{figure}

In the following, we consider (undirected and directed) graphs~$G = (V, E)$ embedded in a surface with non-crossing edges. We require that every face of the embedding is homeomorphic to an open disc, which is called a \emph{cellular} embedding.

Regardless of the (global) orientability of a surface, one can define a local orientation around each node~$v$ of~$G$. If the surface is orientable, these local orientations can be chosen in a way that they are consistent along each edge. In non-orientable surfaces, this is not possible. To keep track of these inconsistencies, one can represent any cellular embedding by an \emph{embedding scheme}~$\Pi = (\pi, \lsignature)$: The \emph{rotation system}~$\pi$ describes, for all nodes, a cyclic permutation of the edges around a node induced by the local orientation. The \emph{signature}~$\lsignature \in \{-1,+1\}^E$ indicates, for every edge, whether the two local orientations (clockwise vs. anti-clockwise) of the adjacent nodes agree ($ +1$) or not ($ -1$), see Figure~\ref{fig:facialWalks}. We assume that an embedded graph is always given together with such an embedding scheme. Conversely, given \emph{any} collection~$\pi$ of cyclic permutations of the edges incident to nodes and \emph{any} vector~$\lsignature \in \{-1,+1\}^E$, there exists a cellular embedding for which~$\Pi = (\pi, \lsignature)$ is a corresponding embedding scheme.

For each face of the embedding of~$G$, let us pick exactly one closed walk along the boundary of this face.
In this way, we obtain a collection of closed walks which we call~\emph{$\Pi$-facial walks} and denote by~$F$.
A more formal definition of~$F$ is provided in Section~\ref{secCharHomology}.
\emph{Euler's Formula} states
\begin{equation}
    \label{eq:Euler}
    \vert V \vert - \vert E \vert + \vert F \vert = 2 - g,
\end{equation}
where~$g$ is the \emph{Euler genus} of the surface. 

Given a graph~$G = (V,E)$ with embedding scheme~$\Pi = (\pi, \lsignature)$ and a set of~$\Pi$-facial walks~$F$, we define a \emph{dual graph}~$G^* = (V^*, E^*)$ as follows: Each node~$f^*$ of~$G^*$ corresponds to a~$\Pi$-facial walk~$f \in F$ of~$G$ and each edge $ e^* \in E^* $ corresponds to an edge $ e \in E $.
If the edge $ e $ is part of two~$\Pi$-facial walks~$f$ and~$g$ in~$G$, the dual edge $ e^* $ is incident to $f^*$ and~$g^*$.
As two $\Pi$-facial walks may share more than one edge and an edge might appear twice in the same $\Pi$-facial walk, the dual graph may have parallel edges and loops. 
If a graph is directed with arc set $ A $, we define the dual graph to be the dual graph of the underlying undirected graph and an arc that corresponds to the dual edge $ e^* $ is called $ e $.

We will equip the dual graph with a \emph{dual embedding scheme}~$\Pi^* = ( \pi^*, \lambda^* )$: the traversing directions of the~$\Pi$-facial walks in~$F$ directly correspond to the dual rotation system~$\pi^*$, and the signature~$\lambda^*(\cdot)$ of a dual edge is positive if the corresponding edge in~$G$ is used in opposite direction by the two corresponding~$\Pi$-facial walks, and negative otherwise. This dual embedding scheme defines an embedding in the same surface. 
The collection~$F^*$ of~$\Pi^*$-facial walks is chosen in a way that their walking directions correspond to the rotation system of~$G$.
An illustration of the dual embedding scheme and its relation to~$\Pi$-facial walks is given in Figure~\ref{fig:Dual}.

\begin{figure}
    \begin{center}
        \includegraphics[width=7cm]{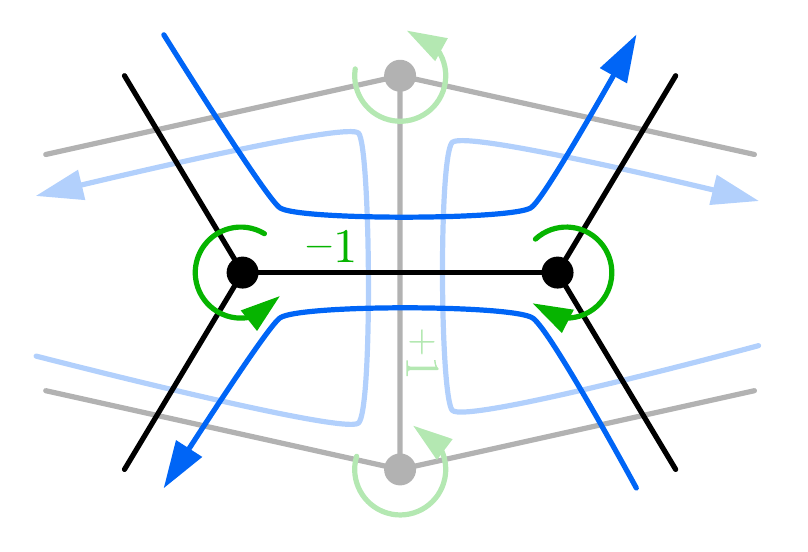}
        \caption{An extract of an embedded graph together with its dual graph (transparent). The embedding scheme and the associated dual embedding scheme are depicted in green. Corresponding facial walks are depicted in blue.}
        \label{fig:Dual}
    \end{center}
\end{figure}

Let~$D = (V, A)$ be a digraph with underlying undirected graph~$G = (V, E )$ and dual graph~$G^*$. 
For any walk~$W = (v_1, e_1, v_2, e_2, \dots, e_{\ell-1}, v_{\ell})$ in~$G$, we define the corresponding \emph{characteristic flow}~$\chi(W) \in \Z^A$ to be an assignment vector on the arcs of~$D$ indicating the total flow over the arcs when sending one unit of flow along~$W$.
This means that for~$(v,w) \in A$,~$\chi(W)((v,w))$ equals the number of appearance of the subsequence~$(v, \{v,w\}, w)$ in~$W$ minus the number of appearance of~$(w, \{v,w\}, v)$.

\subsection{Homology}
\label{subsec:homology}
Given a directed graph~$D = (V, A)$ cellularly embedded in~$\surf$, two integer circulations~$x, y \in \Z^A$ are said to be \emph{$ \Z$-homologous} if~$x - y = \sum_{f \in F} \fav_f \chi(f)$, where~$\fav_f \in \Z$ for each facial walk~$f \in F$.
Let~$\partial \in \Z^{A \times F}$ be the matrix whose columns are the vectors~$\chi(f)$,~$f \in F$. Problem~\ref{probMain} can then be reformulated as
\begin{equation}
    \label{eq:mainProblemIP}
    \min \Set*{ c^\intercal x \ : \ x = y + \partial \fav, \, x \ge \zerovec, \, x \in \Z^A, \, \fav \in \Z^F }.
\end{equation}
For the case of an orientable surface, it is easy to see that the matrix~$\partial$ is totally unimodular~\cite{DeyHK2011}, which implies that Problem~\ref{probMain} can be solved efficiently in this case.
Unfortunately,~$\partial$ is not totally unimodular whenever the surface is non-orientable.

\section{A polynomial-time algorithm on non-orientable surfaces with fixed genus}
\label{sec:fixedgenus_polytimealg}

While it is easy to obtain a polynomial-time algorithm for Problem~\ref{probMain} on orientable surfaces, much more work is required for non-orientable surfaces.
In this section we describe an algorithm that runs in polynomial time for surfaces of \emph{fixed} Euler genus.
We will later see that the problem becomes NP-hard for general surfaces.
A main ingredient of our algorithm is the following characterization of integer circulations that are~$\Z$-homologous to a given one.

\begin{theorem}
    \label{thm:zHomologousNonOrientable}
    Let~$D = (V,A)$ be a digraph cellularly embedded in a non-orientable surface of Euler genus~$g$ and~$y$ an integer circulation in~$D$. 
    Then, there exist vectors~$w_1,\dots,w_{g-1} \in \{0,\pm 1,\pm 2\}^A$,~$h \in \{0,1\}^A$ such that the following holds:
    An integer circulation~$x \in \Z^A$ is~$\Z$-homologous to~$y$, if and only if,
    \begin{equation}
        \label{eq32h9gh9}   
        w_i^\intercal x = w_i^\intercal y \quad \text{ for all } i \in [g-1]
    \end{equation}
    and
    \begin{equation}
        \label{eq4g3894g}
        h^\intercal x \equiv h^\intercal y \pmod* 2.
    \end{equation}
    Moreover,~$w_1,\dots,w_{g-1},h$ can be computed in polynomial time.
\end{theorem}

The constraints in~\eqref{eq32h9gh9} describe the affine subspace of all~$x$ for which there exist~$\fav_f$,~$f \in F$ such that~${x - y = \sum_{f \in F} \fav_f \chi(f)}$.
The parity constraint in~\eqref{eq4g3894g} then characterizes those~$x$, for which the coefficients~$\fav_f$ can be chosen to be integer. A proof of this characterization is given in Section~\ref{secNonOrientable}.

From now on, let us fix a non-orientable surface of Euler genus~$g$ and let~$w_1,\dots,w_{g-1},h$ be given as in Theorem~\ref{thm:zHomologousNonOrientable}.
Recall that in Problem~\ref{probMain} we are given costs~$c \in \R^A_{\ge 0}$ and an integer circulation~$y \in \Z^A$ in~$D$, and we want to find a minimum-cost non-negative integer circulation that is~$\Z$-homologous to~$y$.

In what follows, we will exploit the basic fact that every non-negative circulation can be decomposed into circulations that correspond to directed cycles.
To see whether a sum of such circulations is feasible for Problem~\ref{probMain}, we make use of the following notation.
Let~$d \in \Z^{g-1}$ be the vector whose~$i$-th entry is equal to~$w_i^\intercal y$. Set~$e := h^\intercal y \pmod* 2 \in \{0,1\}$.
For each walk~$W$ in~$D$, consider the vector~$q(W) \in \Z^{g-1}$ whose~$i$-th entry is equal to~$w_i^\intercal \chi(W)$.
Moreover, set~$p(W) := h^\intercal \chi(W) \pmod* 2 \in \{0,1\}$ and~$B := 2|V|$.

A closed walk~$W = v_1,a_1,v_2,\dots,v_{k-1},a_{k-1},v_k$ in~$D$ is called a \emph{$B$-walk} if~$\|q(W_i)\|_\infty \le B$ holds for all subwalks~$W_i = v_1,a_1,v_2,\dots,v_{i-1},a_{i-1},v_i$.
Notice that every directed cycle is a~$B$-walk, and hence every non-negative integer circulation is the sum of circulations that correspond to~$B$-walks.
We consider the set
\[
    \Omega := \{ (q(W), p(W)) \ : \ W \text{ is a $B$-walk in } D\}.
\]
\begin{lemma}
    \label{lem2083}
    For each~$(q,p) \in \Omega$, one can compute in polynomial time a~$B$-walk~$W =: W_{q,p}$ in~$D$ with~$q(W) = q$ and~$p(W) = p$ that minimizes~$c^\intercal \chi(W)$.
\end{lemma}
For the sake of exposition, we provide a proof at the end of this section. Notice that~$|\Omega| \le 2 (2B + 1)^{g-1} = \mathrm{poly}(|V|)$, hence the collection~$\{ W_{q,p} : (q,p) \in \Omega \}$ can be computed in polynomial time. Let us now consider the following set
\begin{align*}
    \cc := \bigg \{ \sum \nolimits_{(q,p) \in \Omega} z_{q,p} \chi(W_{q,p}) \ : \ & z_{q,p} \in \Z_{\ge 0} \text{ for every } (q,p) \in \Omega, \\[-1em]
    & \sum \nolimits_{(q,p) \in \Omega} z_{q,p} q = d, \, \sum \nolimits_{(q,p) \in \Omega} z_{q,p} p \equiv e \pmod* 2 \bigg \}
\end{align*}
of non-negative integer circulations in~$D$.
\begin{lemma}
    \label{lem3877}
    Every circulation in~$\cc$ is feasible for Problem~\ref{probMain}.
    Moreover,~$\cc$ contains at least one optimal solution.
\end{lemma}
Again, we postpone the proof to the end of this section.
Setting~$\tilde{c}_{q,p} := c^\intercal \chi(W_{q,p})$ for each~$(q,p) \in \Omega$, by Lemma~\ref{lem3877} it remains to obtain a solution for
\begin{align*}
    & \min \bigg \{ \sum \nolimits_{(q,p) \in \Omega} \tilde{c}_{q,p} z_{q,p} \ : \
        z \in \Z^\Omega_{\ge 0}, \,
        \sum \nolimits_{(q,p) \in \Omega} z_{q,p} q = d, \,
        \sum \nolimits_{(q,p) \in \Omega} z_{q,p} p \equiv e \pmod* 2 \bigg \} \\
    = & \min \bigg \{ \sum \nolimits_{(q,p) \in \Omega} \tilde{c}_{q,p} z_{q,p} \ : \
        z \in \Z^\Omega_{\ge 0}, \,
        k \in \Z_{\ge 0}, \,
        \sum \nolimits_{(q,p) \in \Omega} z_{q,p} q = d, \,
        \sum \nolimits_{(q,p) \in \Omega} z_{q,p} p = 2k + e \bigg \}.
\end{align*}
Notice that the latter is an integer program in~$n := |\Omega| + 1 = \mathrm{poly}(|V|)$ variables of the form
\[
    \min \left\{ \bar{c}^\intercal x \ : \ \bar{A}x = \bar{b}, \, x \in \Z^n_{\ge 0} \right\},
\]
where~$\bar{A} \in \Z^{g \times n}$ and~$\bar{b} \in \Z^g$.
Recall that the entries in~$\bar{A}$ are polynomially bounded in~$n$, and that~$g$ (the number of rows in~$\bar{A}x = \bar{b}$) is assumed to be fixed.
It is known that integer programs of this form can be solved in polynomial time.
For instance, a polynomial-time algorithm for this setting is described in~\cite{artmann2016note}, which is based on Papadimitriou's pseudopolynomial-time algorithm for integer programs with a fixed number of constraints~\cite{papadimitriou1981complexity}.
Another approach can be found in~\cite[Thm.~3.3]{eisenbrand2019proximity}.
This finishes the proof of Theorem~\ref{thm:existenceOfPolyAlgo}.
We close this section by providing the proofs for Lemma~\ref{lem2083} and Lemma~\ref{lem3877}.
\begin{proof}[Proof of Lemma~\ref{lem2083}]
    We determine each~$W_{q,p}$ by computing shortest paths in the following auxiliary graph~$\widebar{D} = (\widebar{V}, \widebar{A})$ defined by
 	\begin{align*}
	\widebar{V} &:= \bigg \{ (v,(x,y)) \ : \ v \in V, \ x \in \{-B, \ldots, B\}^{g-1}, y \in \Z_2 \bigg \},\\
 	\widebar{A} &:= \Bigg \{ \left( (v,(x,y)), (v',(x',y')) \right) \ : \ 
 	\begin{array}{@{}l@{}}
 		(v,(x,y)), (v',(x',y')) \in \widebar{V}, \ (v,v') \in A, \\
 		x + M \chi((v,v')) = x',\\
 		y + \ h^\intercal \chi((v,v'))\equiv y' \pmod* 2
 	\end{array} \Bigg \}.
 	\end{align*}
 	Here, $M$ is the matrix whose rows are the vectors~$w_1^\intercal,\dots,w_{g-1}^\intercal$.
 	Observe that for every walk~$W$ we have~$M \chi(W) = q(W)$.
    The cost~$\widebar{c}$ of an arc~$\widebar{a} = \left( (v,(x,y)), (v',(x',y')) \right)$ in~$\widebar{A}$ is defined by~$\widebar{c}(\widebar{a}) := c((v,v'))$.
    Notice that~$\widebar{D}$ can be constructed in polynomial time and that~$\widebar{c}$ is non-negative.

    Let~$(q,p) \in \Omega$ and fix a node~$v \in V$.
    We observe that there is a bijection between~$B$-walks~$W$ in~$D$ starting (and ending) at~$v$ with~$q(W) = q$ and~$p(W) = p$, and walks~$\widebar{W}$ in~$\widebar{D}$ from~$(v, (\zerovec, 0))$ to~$(v,(q(W), p(W)))$.
    Moreover, the costs of~$W$ and~$\widebar{W}$ coincide. Indeed, let~$W$ be a~$B$-walk in~$D$ with~$q(W) = q$ and~$p(W) = p$, which starts at~$v$.
    Let~$v_1,\dots,v_k,v_1$ be the sequence of nodes visited by~$W$, and let~$W_i$ denote the respective subwalk from~$v_1$ to~$v_i$.
    Then, walk~$\widebar{W}$ in~$\widebar{D}$ is obtained by visiting the nodes
    \[
        (v_1, (\zerovec, 0)), \, (v_2, (q(W_2), p(W_2))), \dots, (v_k, (q(W_k), p(W_k))), \, (v_1, (q(W), p(W)))
    \]
    in the given order, and the cost of~$W$ equals the cost of~$\widebar{W}$.

    Conversely, consider any walk~$\widebar{W}$ from~$(v,(\zerovec,0))$ to~$(v,(q,p))$ in~$\widebar{D}$, and let~$(v_i, (x_i, y_i))$,~$i=1,\dots,k$, be the sequence of nodes it visits.
    Define~$W$ to be the closed walk that visits the nodes~$v_1,\dots,v_k$.
    We see that~$W$ is a~$B$-walk with~$q(W) = q$ and~$p(W) = p$, and that the costs of~$\widebar{W}$ and~$W$ coincide.

    We conclude that a~$B$-walk~$W$ in~$D$ with~$q(W) = q$ and~$p(W) = p$ minimizing~$c^\intercal \chi(W)$ can be found by computing a shortest path in~$\widebar{D}$ from~$(v, (\zerovec, 0))$ to~$(v,(q(W), p(W)))$ for every~$v \in V$, and returning the walk in~$D$ that corresponds to the path of minimum length.
\end{proof}
\begin{proof}[Proof of Lemma~\ref{lem3877}]
    Again, let~$M$ be the matrix whose rows are the vectors~$w_1^\intercal,\dots,w_{g-1}^\intercal$.
    Observe that for every walk~$W$ we have~$M \chi(W) = q(W)$.
    Moreover, by Theorem~\ref{thm:zHomologousNonOrientable} a non-negative integer circulation~$x$ in~$D$ is feasible for Problem~\ref{probMain}, if and only if,~$Mx = d$ and~$h^\intercal x \equiv e \pmod* 2$.

    Let~$x = \sum \nolimits_{(q,p) \in \Omega} z_{q,p} \chi(W_{q,p})$ be any circulation in~$\cc$.
    First, notice that each~$\chi(W_{q,p})$ is a non-negative integer circulation in~$D$, and so is~$x$.
    Moreover, we have
    \[
        Mx = \sum \nolimits_{(q,p) \in \Omega} z_{q,p} M \chi(W_{q,p}) = \sum \nolimits_{(q,p) \in \Omega} z_{q,p} q(W_{q,p}) = \sum \nolimits_{(q,p) \in \Omega} z_{q,p} q = d
    \]
    as well as
    \[
        h^\intercal x = \sum \nolimits_{(q,p) \in \Omega} z_{q,p} h^\intercal \chi(W_{q,p}) \equiv \sum \nolimits_{(q,p) \in \Omega} z_{q,p} p \equiv e \pmod* 2.
    \]
    Thus,~$x$ is feasible for Problem~\ref{probMain}.
    
    Now let~$x^*$ be an optimal solution to Problem~\ref{probMain}.
    As discussed earlier, we may decompose~$x^*$ into~$B$-walks~$W_1,\dots,W_k$ such that~$x^* = \sum_{i=1}^k z_i \chi(W_i)$, where~$z_1,\dots,z_k \in \Z_{\ge 0}$.
    Clearly, we have that~$(q_i, p_i) := (q(W_i), p(W_i)) \in \Omega$ for~$i=1,\dots,k$.
    Consider the non-negative integer circulation
    \[
        x' := \sum \nolimits_{i=1}^k z_i \chi(W_{q_i,p_i}).
    \]
    As~$x^*$ is feasible, we have~$Mx^* = d$ and~$h^\intercal x^* \equiv e \pmod* 2$, which yields
    \[
        \sum \nolimits_{i=1}^k z_i q_i = \sum \nolimits_{i=1}^k z_i M \chi(W_i) = Mx^* = d
    \]
    as well as
    \[
        \sum \nolimits_{i=1}^k z_i p_i = \sum \nolimits_{i=1}^k z_i h^\intercal \chi(W_i) = h^\intercal x^* \equiv e \pmod* 2.
    \]
    This shows that~$x' \in \cc$.
    In particular,~$x'$ is also feasible for Problem~\ref{probMain}.
    By definition of~$W_{q_i, p_i}$, we have~$c^\intercal \chi(W_{q_i, p_i}) \le c^\intercal \chi(W_i)$ for~$i=1,\dots,k$, which yields~$c^\intercal x' \le c^\intercal x^*$.
    Therefore,~$x'$ is also an optimal solution to Problem~\ref{probMain}.
\end{proof}

\section{Hardness for instances on general non-orientable surfaces}
\label{sec:hardness}

In the previous section, we have shown that Problem~\ref{probMain} can be solved in polynomial time on non-orientable surfaces of fixed Euler genus. This problem becomes NP-hard on general non-orientable surfaces.

Let us consider the following problem, which is a special case of (the decision version of) Problem~\ref{probMain}.
\begin{prob}
    \label{hardProblem}
        Given a digraph~$D = (V, A)$ cellularly embedded in a surface with arc cost~$c \in  \{0,\frac{1}{2}, 1\}^A$ such that~$\onevec \in \Z^A$ is a circulation in~$D$, and an integer~$k$,
        decide whether there exists a non-negative integer circulation in~$D$ that is~$\Z$-homologous to~$\onevec$ and has cost at most~$k$.
\end{prob}
In what follows, we will prove that Problem~\ref{hardProblem} is NP-hard, which implies Theorem~\ref{thm:hard}. We will also see that the problem remains hard if we restrict ourselves to circulations in~$\{0,1\}^A$.
In Section~\ref{se:ReductionSTABtoCIRC}, we show that the following problem can be efficiently reduced to Problem~\ref{hardProblem}.

\begin{prob}
	\label{prob:Stab_Z}
	Given a connected graph~$G = (V, E)$ together with edge cost~$c \in  \{0,\frac{1}{2}, 1\}^E$, and an integer~$k$, decide whether there exists a vector~$x \in \Z^V$ satisfying~$x(v) + x(w) \le 1$ for each~$\{v,w\} \in E$ and~$\sum_{\{v,w\} \in E} c(\{v,w\}) (x(v) + x(w)) \ge k$.
\end{prob}

Problem~\ref{prob:Stab_Z} can be seen as a special stable set problem where we neglect the non-negativity constraints.
We will show that the following special case of the weighted stable set problem can be efficiently reduced to Problem~\ref{prob:Stab_Z}.
A proof is given in Section~\ref{sec23789}.
The node weights in Problem~\ref{prob:specialSTAB} and~\ref{prob:Stab_Z} are induced by edge costs, meaning that the weight of a node is just the sum of the costs of its incident edges.

\begin{prob}
    \label{prob:specialSTAB}
    Given a graph~$G = (V, E)$ with edge cost~$c \in \{0,\frac{1}{2}, 1\}^E$ and an integer~$k$, decide whether there exists a stable set~$S \subseteq V$ in~$G$ such that~$\sum_{e \in E} c(e) |S \cap e| \ge k$.
\end{prob}

Finally, we show that Problem~\ref{prob:specialSTAB} is NP-hard by a reduction from 3-SAT in the next section. This concludes the proof of Theorem~\ref{thm:hard}.

\subsection{Hardness of Problem~\ref{prob:specialSTAB}}
\label{sec:hardnessofspecialSTABproblem}
The following reduction is based on the standard reduction for the classical stable set problem, see Garey and Johnson~\cite{garey1979computers}.

Let~$(U, C)$ be any instance of~$3$-SAT, where~$U$ is the set of variables and~$C$ denotes the set of clauses.
Now, for each variable~$u \in U$ we define the graph~$G^u$ consisting of two nodes representing~$u$ and its negation~$\bar{u}$, which are joined by an edge~$e_u$.
The cost of this edge is set to~$c(e_u) := 1$.
Next, for each clause~$c \in C$ we define a triangle graph~$G^c$ containing one node for each literal in~$c$ and three edges connecting them.
We assign a cost of~$\frac{1}{2}$ to all edges in the triangle.
Finally, we define~$G = (V,E)$ as union of all~$G^u$ ($ u \in U$) and~$G^c$ ($c \in C$) together with the following additional edges:
For each literal~$\ell$ that appears in a clause~$c$ and corresponds to variable~$u$, connect the node in~$G^c$ that represents~$\ell$ with the node in~$G^u$ that represents the negation of~$\ell$.
The edge cost~$c$ for all these additional edges is defined to be zero.

Notice that every stable set~$S$ in~$G$ satisfies~$|S| = \sum_{e \in E} c(e) |S \cap e|$.
We leave to the reader to check that~$(U, C)$ is satisfiable if and only if~$G$ has a stable set~$S$ of cardinality~$|S| \ge |U| + |C|$.
A formal proof can be found in~\cite{garey1979computers}.

\subsection{Reduction from Problem~\ref{prob:specialSTAB} to Problem~\ref{prob:Stab_Z}}
\label{sec23789}
First, we  may assume~$G$ to be connected, otherwise we would treat each component separately. Since the problem becomes easy in the case of bipartite graphs, let~$G$ be non-bipartite.

It remains to show that~$G$ has a stable set~$S$ with~$\sum_{e \in E} c(e) |S \cap e| \ge k$, if and only if, there exists a vector~$x \in \Z^V$ satisfying~$x(v) + x(w) \le 1$ for each~$\{v,w\} \in E$ and
\begin{equation}
    \label{eq3948}
    \sum_{\{v,w\} \in E} c(\{v,w\}) (x(v) + x(w)) \ge k.
\end{equation}
If~$G$ has a stable set~$S$ with~$\sum_{e \in E} c(e) |S \cap e| \ge k$, we define~$x \in \{0,1\}^V$ to be the characteristic vector of~$S$.
For each edge~$e = \{v,w\} \in E$ we clearly have~$|S \cap e| = x(v) + x(w) \le 1$, and hence~\eqref{eq3948} holds.

Conversely, suppose that
\[
    \max \left \{ \sum \nolimits_{\{v,w\} \in E} c(\{v,w\}) (x(v) + x(w)) \ : \ x \in \Z^V, \, x(v) + x(w) \le 1 \text{ for all } \{v,w\} \in E \right \} \ge k.
\]
Since~$G$ is non-bipartite, it can be shown that the convex hull of feasible solutions to the above integer program is a pointed polyhedron.
Moreover, it can be shown that each vertex of this polyhedron is a~$0/1$-vector.
Both facts and their proofs can be found in~\cite[Proposition~12]{ConfortiFHJW2019}.
Moreover, as~$c$ is non-negative, the above integer program is certainly bounded.
This means that the optimum to the above integer program is attained at a point~$x \in \{0,1\}^V$, which yields the claim.

\subsection{Reduction from Problem~\ref{prob:Stab_Z} to Problem~\ref{hardProblem}}
\label{se:ReductionSTABtoCIRC}
The following reduction is based on methods developed in~\cite{ConfortiFHJW2019} that are designed for graphs with a particular embedding. First, we have to construct such an embedding for the input graph~$G$ which will have the property that the edges of the dual graph~$G^*$ may be directed such that every facial walk in~$G^*$ is a directed walk.

The embedding is obtained by equipping~$G$ with a rotation system~$\Pi = (\pi, \lsignature)$ in which the signature of every edge is defined to be~$-1$. Its cyclic permutations around each node can be chosen arbitrarily.

We may define the set~$F$ of~$\Pi$-facial walks in~$G$ corresponding to this embedding. It can then be used to define the dual graph~$G^*$ together with a corresponding dual embedding scheme~$\Pi^* = (\pi^*, \lambda^*)$ and a set~$F^*$ of exactly these~$\Pi^*$-facial walks whose traversing directions correspond to the cyclic permutations~$\pi$ around the nodes in~$G$. 
Since~$\lsignature = -\onevec$, all dual edges in~$G^*$ will always be used in the same direction in the~$\Pi^*$-facial walks in~$F^*$.
To obtain a digraph~$D = (V^*, A)$ we direct all edges in~$G^*$ corresponding to the direction in which the edges are used in the walks in~$F^*$. Now, all~$\Pi^*$-facial walks in~$F^*$ are directed walks.
Observe that in~$D$ the vector~$\onevec \in \Z^A$ is a circulation. Indeed, it is half times the sum over all characteristic flows of facial walks in~$F^*$, i.e.~$\onevec = \frac{1}{2} \sum_{v^* \in F^*} \chi(v^*)$.
Since there is a one-to-one correspondence between~$E$ and~$A$, the cost on the edges in~$G$ may also be seen as arcs cost~$c \in \{0,\frac{1}{2}, 1\}^A$. 
It remains to show that there exists a vector~$x \in \Z^V$ satisfying~$x(v) + x(w) \le 1$ for each~$\{v,w\} \in E$ and
\begin{equation}
    \label{eq1238}
    \sum_{\{v,w\} \in E} c(\{v,w\}) (x(v) + x(w)) \ge k,
\end{equation}
if and only if, there exists a non-negative integer circulation in~$D$ that is~$\Z$-homologous to~$\onevec \in \Z^A$ and has cost at most~$\sum_{e \in E} c(e) - k$.

Suppose that there exists a vector~$x \in \Z^V$ satisfying~$x(v) + x(w) \le 1$ for each~$\{v,w\} \in E$ and Inequality~\eqref{eq1238}.
Since each facial walk in~$G^*$ is a directed walk, 
\[
    y := \onevec - \sum_{v \in V} x(v) \chi(v^*)
\]
is an integer circulation in~$D$ that is~$\Z$-homologous to~$\onevec$.
Moreover, each arc~$a$ in~$D$ appears exactly twice in~$\Pi^*$-facial walks in~$F^*$ and the nodes (or node) in~$G$ corresponding to these~$\Pi^*$-facial walks are joined by an edge whose dual edge corresponds to~$a$.
Since~$x(v) + x(w) \le 1$ for each~$\{v,w\} \in E$, the value of~$y$ assigned to the dual arc corresponding to~$\{v,w\}$ is non-negative.
Moreover, we have
\begin{align*}
    c^\intercal y
    = c^\intercal \left(\onevec - \sum \nolimits_{v \in V} \chi(v^*) x(v) \right)
    & = \sum_{e \in E} c(e) - \sum_{v \in V} \left(\sum \nolimits_{e \in \delta(v)} c(e)\right) x(v) \\
    & = \sum_{e \in E} c(e) - \sum_{\{v,w\} \in E} c(\{v,w\}) (x(v) + x(w)) \\
    & \le \sum_{e \in E} c(e) - k.
\end{align*}

Conversely, consider any non-negative integer circulation~$y$ in~$D$ that is~$\Z$-homologous to the circulation~$\onevec$ with cost at most~$\sum_{e \in E} c(e) - k$.
Since~$y$ is~$\Z$-homologous to~$\onevec$, there exist coefficients~$\fav_{v^*}$ for all~$v^* \in \Z^{F^*}$ such that~$y = \onevec - \sum_{v^* \in F^*} \fav_{v^*} \chi(v^*)$.
As there is a on-to-one correspondence between nodes~$v$ in~$G$ and~$\Pi^*$-facial walks~$v^*$ in~$F^*$, we may define the vector~$x \in \Z^V$ via~$x(v) = \fav_{v^*}$, for all~$v \in V$.
Since~$y$ is non-negative, the sum of two coefficients~$\fav_{v^*}$ and~$\fav_{w^*}$ that correspond to facial walks using the same arc can never exceed one.
Moreover, we have
\begin{align*}
    \sum_{\{v,w\} \in E} c(\{v,w\}) (x(v) + x(w))
    & = \sum_{v \in V} \left(\sum \nolimits_{e \in \delta(v)} c(e)\right) x(v)
    = c^\intercal \sum_{v \in V} x(v) \chi(v^*) \\
    & = \sum_{e \in E} c(e) - c^\intercal \left(\onevec - \sum \nolimits_{v \in V} x(v) \chi(v^*)\right) \\
    & = \sum_{e \in E} c(e) - c^\intercal y \ge k,
\end{align*}
which concludes the proof.

\section{Characterizing homology}
\label{secCharHomology}

In this section, we present alternative characterizations of homology, leading to a discussion of Theorem~\ref{thm:descriptionOfP(D,y)} and the proof of Theorem~\ref{thm:zHomologousNonOrientable}.
To this end, we need to provide further details and definitions regarding surfaces and homology.
For further information, we refer to the books of Hatcher~\cite{Hatcher2005} and Mohar and Thomassen~\cite{MoharT2001}.

Consider a graph~$G = (V,E)$ cellularly embedded in a surface with a corresponding embedding scheme~$\Pi = (\pi, \lsignature)$. In what follows, we provide a more formal definition of the set~$F$ of~$\Pi$-facial walks.

Consider the following procedure which defines a~$\Pi$-facial walk, see Figure~\ref{fig:facialWalks}: Start at an arbitrary node~$v$ and an edge~$e$ incident to~$v$, then traverse~$e$ and continue the walk at the edge~$e^{\prime}$ coming after, or before,~$e$ in the cyclic permutation given by~$\pi$ if the signature of~$e$ is positive, or negative, respectively. 
Reaching the next node, we continue again with the edge coming after, or before, the edge~$e^{\prime}$ if the number of already traversed edges with negative signature is even, or odd, respectively. We continue until the following three conditions are met: (i) we reach the starting node~$v$, (ii) the number of traversed edges with negative signature is even, (iii) the next edge would be the starting edge~$e$. In this way, we obtain a collection of closed walks which we then call~$\Pi$-facial walks. Notice that the~$\Pi$-facial walks of a digraph are walks in the underlying undirected graph. We consider two~$\Pi$-facial walks to be equivalent if they only differ by a cyclic shift of nodes and edges or if one is the reverse of the other one. Let us pick one~$\Pi$-facial walk from each equivalence class and denote the resulting set of walks by~$F$. Notice that every edge is either contained twice in one walk in~$F$ or in exactly two walks in~$F$. 

\begin{figure}[htb]
	\centering
	\includegraphics[width=0.5\textwidth]{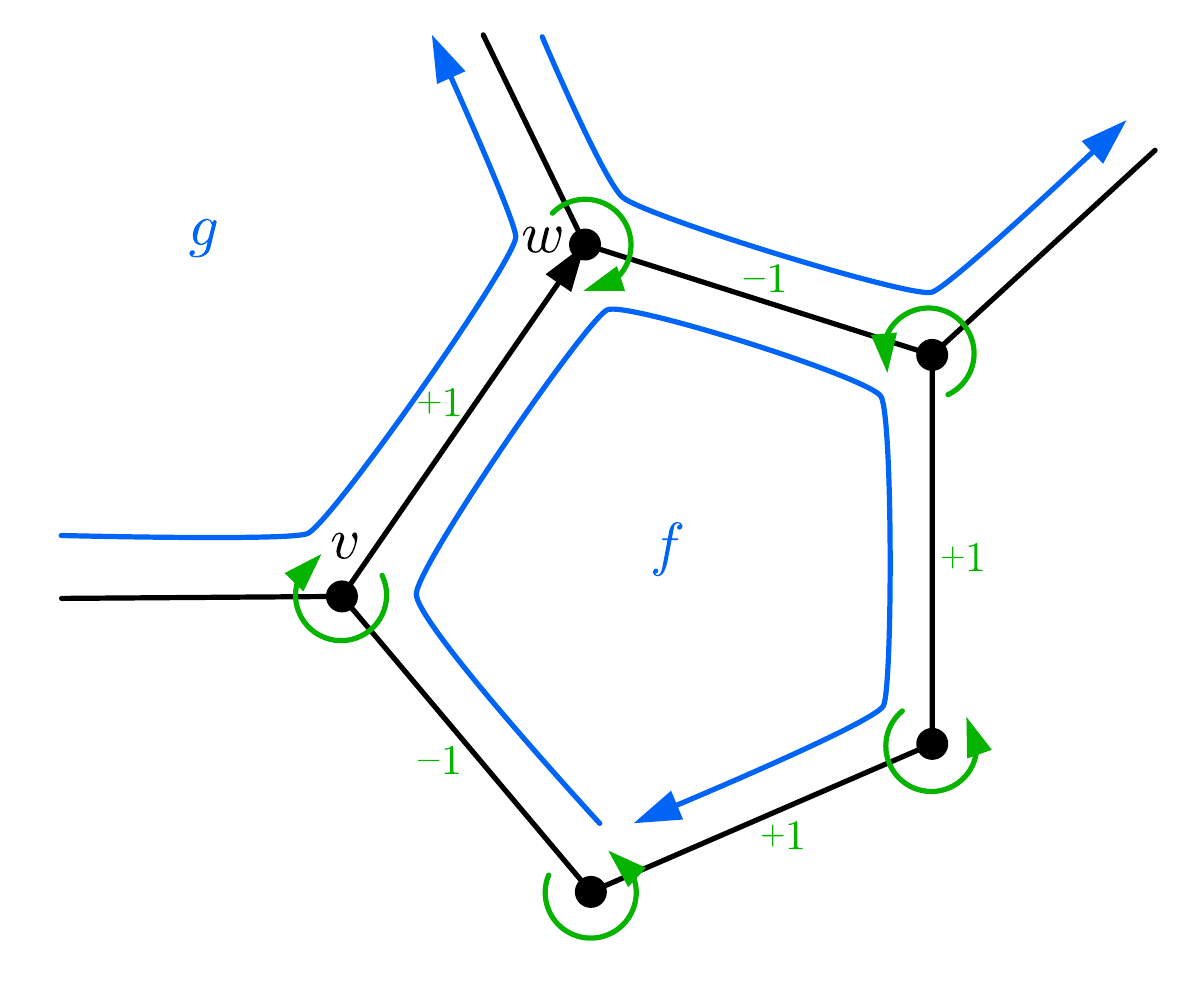}
	\caption{An extract of an embedded graph. The embedding scheme is depicted in green: Arrows around the nodes indicate their local orientations and the numbers on the edges the induced signature. The facial walks in~$F$ are drawn in blue.}
	\label{fig:facialWalks}
\end{figure}

By construction, the number of used edges with negative signature in a~$\Pi$-facial walk is even.
If a closed walk in~$G$ traverses an even number of edges with negative signature, it is called \emph{two-sided}, otherwise it is called \emph{one-sided}.
It turns out that every two-sided cycle in the surface has a neighborhood that is homeomorphic to an annulus, whereas every one-sided cycle has a neighborhood that is homeomorphic to an open Möbius band. It follows that the presence of a one-sided cycle implies that the underlying surface is non-orientable.

To elaborate on an alternative characterization of homology, we also need the following notion.
In addition to the notion of the characteristic flow~$\chi(W) \in \Z^A$ for walks~$W$ in~$G$, we define the vector~$\xi(H^*) \in \Z^A$ for any walk~$H^*$ in the dual graph~$G^*$ as follows. Intuitively, we think of~$\xi(H^*)$ as a flow that sends one unit along the edges in~$H^*$. Whenever a unit is sent along a dual edge, we account it for the corresponding arc in~$D$. The sign of this value will depend on the direction we traverse~$H^*$ along this arc. Formally, consider any arc~$a = (v, w) \in A$ and let~$f$ be any $\Pi$-facial walk in~$F$ of~$G$. Set~$s(a, f) \in \{-1,0,1\}$ to be non-zero in the case that edge~$\{v,w\}$ appears once $f$ or twice in the same direction and zero otherwise. If~$f$ traverses the edge from~$v$ to~$w$, then~$s(a, f) = 1$, otherwise~$s(a, f) = - 1$. Observe, that the sign of $ s(a,f) $ equals the sign of $ \partial_{a,f} $. For instance, in Figure~\ref{fig:facialWalks},~$s((v,w),f) = 1$. 
Now, for a walk~$H^* = (f_1^*, e_1^*, f_2^*, \dots, e_{\ell-1}^*, f_{\ell}^*)$ in~$G^*$ and arc~$a \in A$, we define
\begin{equation*}
    \xi(H^*)(a) := \sum_{\substack{i \in \{1, \dots, \ell -1\} \\ a = e_i}} \lambda^*(e_1^*)\cdot\cdots\cdot \lambda^*(e_{i-1}^*) s(a, f_i).
\end{equation*} 
Observe that~$\langle z, \xi(v^*) \rangle = 0$ for any circulation~$z$ in~$D$ and any~$\Pi^*$-facial walk~$v^*$ in~$G^*$.

Before we start with the characterization of~$\Z$-homology, we consider the slightly weaker concept of \emph{$\R$-homology}, which arises by dropping the integrality condition for the~$\fav_f$ in the definition of~$\Z$-homology.
More formally, given a directed graph~$D = (V, A)$ cellularly embedded in~$\surf$, two circulations~$x,y \in \R^A$ are called~$\R$-homologous if~$x - y$ is a linear combination of characteristic flows of~$\Pi$-facial walks, which we also call \emph{facial circulations}. That is, there exists an assignment vector~$\fav \in \R^F$ with a coefficient~$\fav_f \in \R$ for each facial walk~$f \in F$, such that~$x - y = \sum_{f \in F} \fav_f \chi(f)$. 
To rewrite the above in a compact way, recall the matrix~$\partial = \partial_D \in \Z^{A \times F}$ defined by
\begin{equation*}
    \partial_{a,f} := \chi(f)(a) \quad \text{ for all } a \in A, \, f \in F.
\end{equation*}
Circulations~$x, y$ are~$\R$-homologous if~$x = y + \partial \fav$, for some~$\fav \in \R^F$. 

Observe that two circulations~$x, y$ are~$\R$-homologous ($ \Z$-homologous), if and only if~$x-y$ is~$\R$-homologous ($ \Z$-homologous) to the circulation~$\zerovec \in \R^A$. 
For this reason, in what follows we will first provide an alternative description of circulations that are~$\R$-homologous ($\Z$-homologous) to~$\zerovec \in \R^A$, which then directly yields characterizations for~$\R$-homology ($ \Z$-homology) between two arbitrary circulations.

\subsection{Orientable surfaces}
\label{secOrientable}

For any digraph~$D$, we denote by~$P(D)$ the convex hull of non-negative integer circulations in~$D$. It is a basic fact that~$P(D)$ is actually equal to the set of all non-negative circulations in~$D$. Hence, this polyhedron can be described as the set of all~$x \in \R^A_{\ge 0}$ that satisfy the ``flow conservation'' constraints.
Regarding Problem~\ref{probMain}, we are interested in the convex hull of only those integer circulations in~$P(D)$ that are~$\Z$-homologous to a given integer circulation~$y$, and we denote the respective polyhedron by~$P(D, y)$. The purpose of this section is to show that a description of~$P(D, y)$ can be easily obtained in the orientable case. 

As mentioned in Section~\ref{subsec:homology}, matrix~$\partial$ is totally unimodular in the case of orientable surfaces, and hence, Problem~\ref{probMain} can be solved in polynomial time.
Here, we would like to elaborate on another consequence for the description of~$P(D, y)$. By expressing~$\Z$-homology using~$\partial$, we know that
\begin{align*}
P(D, y) &= \conv \Set*{ x \in \Z^A \ : \ x = y + \partial \fav, \,  x \ge \zerovec, \, \fav \in \Z^F } \\
&= \conv \Set*{ x \in \R^A \ : \ x = y + \partial \fav, \, x \ge \zerovec, \, \fav \in \R^F },
\end{align*}
where the second equality follows from the integrality of the latter polyhedron, a consequence of~$\partial$ being totally unimodular. This means that~$P(D, y)$ is the set of all non-negative circulations in~$D$ that are~$\R$-homologous to~$y$. Denoting by~$L(D, y)$ the set of all circulations in~$D$ that are~$\R$-homologous to~$y$, we obtain
\[
P(D, y) = P(D) \cap L(D, y),
\]
and hence Theorem~\ref{thm:descriptionOfP(D,y)}. 
To obtain an even more explicit description of~$P(D, y)$, observe that~$L(D, y)$ is an affine subspace which is generated by all facial circulations and shifted by~$y$. 
First, let us consider the case in which~$y = \zerovec$. The set~$L(D, \zerovec)$ of all circulations in~$D$ that are~$\R$-homologous to~$\zerovec$ is the subspace generated by all facial circulations.
If the surface is orientable, the facial circulations generate a space of dimension~$|F| - 1$. On the other hand, it is well-known that the space of all circulations in~$D$ is~$(|A| - |V| + 1)$-dimensional. Thus, besides the constraints describing the set of all circulations, Euler's formula~\eqref{eq:Euler} yields that we need~$g$ additional constraints to obtain~$L(D, \zerovec)$.

These constraints can be obtained by the following construction, also see~\cite{ChambersEN2012}. Pick any spanning tree~$K$ in~$G$ and observe that~$G^* \setminus K^*$ is still connected. Hence, there exists a spanning tree~$T^*$ in~$G^* \setminus K^*$. By Euler's formula, there exist exactly~$g$ edges~$e_1, \dots, e_g$ in~$G$ that are not contained in~$K$ and whose dual edges~$e_1^*, \dots, e_g^*$ are not contained in~$T^*$. For each~$i \in [g]$, we define the cycle~$C_i$ as the unique (dual) cycle in~$T^* \cup \{ e_i^* \}$. These~$g$ cycles will yield the additional constraints needed to describe~$L(D, \zerovec)$.

\begin{proposition}
	\label{propDescriptionL(D,y)Orientable}
	Let~$D$ be a digraph cellularly embedded in an orientable surface of Euler genus~$g$, and let~$C_1, \dots, C_g$ be the (dual) cycles defined above.
	Then,
	\begin{equation*}
	L(D,\zerovec) = \Set*{ x \in \R^A \ : \ x \textnormal{ is a circulation and } \langle x, \xi(C_i) \rangle = 0 \: \forall  i \in [g] }.
	\end{equation*}
\end{proposition}

\begin{proof}
	Let~$L$ denote the linear subspace on the right-hand side. 
	
	We first show that~$L(D, \zerovec) \subseteq L$. Every~$x \in L(D, \zerovec)$ is of the form~$x = \sum_{f \in F} \fav_f \chi(f)$ for some coefficients~$\fav_f \in \R$.
	For every cycle~$H = (f_1^*, f_2^*, \dots, f_k^*)$ in the dual graph~$G^*$, we have
	\[\langle x, \xi(H) \rangle = \sum_{i = 1}^k (\fav_{f_i} - \fav_{f_{i+1}}) = 0,\]
	where~$f_{k+1} = f_1$. This shows that~$x \in L$, and hence~$L(D, \zerovec) \subseteq L$.
	
	It remains to show that~$\dim(L) \le \dim(L(D, \zerovec))$. Recall that~$\dim(L(D, \zerovec)) = |F| - 1$ and that the space of all circulations has dimension~$|F| - 1 + g$. With each constraint~$\langle x, \xi(C_i) \rangle = 0$ that we iteratively add to the space of all circulations, the dimension drops by one. Indeed, for each~$i \in [g]$ there is a (unique) cycle~$H_i$ in~$K \cup \{e_i\}$. For this~$H_i$, the circulation~$\chi(H_i)$ satisfies all constraints~$\langle \chi(H_i), \xi(C_j) \rangle = 0$, for~$j \ne i$, but~$\langle \chi(H_i), \xi(C_i) \rangle \ne 0$. This means that~$\dim(L) \le (|F| - 1 + g) - g = |F| - 1 = \dim(L(D, \zerovec))$.
\end{proof}

\begin{corollary}
	\label{cor:zhomologOrientableSurface}
	Let~$D = (V,A)$ be a digraph cellularly embedded in an orientable surface of Euler genus~$g$ and let~$y$ be an integer circulation in~$D$. 
	Then, we can efficiently compute cycles~$C_1, \dots, C_g$ in the dual graph of~$D$ such that the following holds:
	An integer circulation~$x \in \Z^A$ is~$\Z$-homologous to~$y$, if and only if,
	\[
	\langle x, \xi(C_i) \rangle = 	\langle y, \xi(C_i) \rangle  \qquad \textnormal{ for all } i \in [g].
	\]
\end{corollary}

\subsection{Non-orientable surfaces}
\label{secNonOrientable}

In this section, we provide a characterization of~$\Z$-homology that is exploited in Section~\ref{sec:fixedgenus_polytimealg}.
To this end, we first provide a characterization of~$\R$-homology similar to Corollary~\ref{cor:zhomologOrientableSurface} in the orientable case.
In the previous section we have seen that for an orientable surface two integer circulations are~$\Z$-homologous, if and only if, they are~$\R$-homologous.
Unfortunately, this is not true for non-orientable surfaces.
Therefore, some more arguments are required to obtain Theorem~\ref{thm:zHomologousNonOrientable}.

Let~$D$ be a digraph cellularly embedded in a non-orientable surface of Euler genus~$g$ according to an embedding scheme~$\Pi = (\pi, \lsignature)$. Let~$G^*$ be the dual graph canonically embedded in the same surface. In order to obtain a description as in Proposition~\ref{propDescriptionL(D,y)Orientable}, let us consider the following construction of~$g-1$ closed walks in~$G^*$. This construction is similar to the construction for embeddings in orientable surfaces described in Section~\ref{secOrientable}, which can already be found in~\cite{ChambersEN2012}.

Pick any spanning 1-tree~$T^*$ in~$G^*$ (a spanning tree with one additional edge forming one single cycle) whose cycle~$C^*$ is a one-sided cycle. Denote the set of arcs in~$D$ containing all arcs whose dual edges are in~$T^*$ by~$T$. Notice that in most cases $T$ is not a tree. Since~$C^*$ is a one-sided cycle,~$D \setminus T$ is still connected. Hence, there exists a spanning tree~$K$ in~$D \setminus T$. 

By Euler's formula~\eqref{eq:Euler} there exist exactly~$g-1$ arcs~$\barcs_1, \dots \barcs_{g-1}$ in~$D$ that are not contained in~$K \cup T$. For each~$i \in [g-1]$, we define~$W^*_i$ as the (dual) two-sided closed walk in~$T^* \cup \{ \barcs_i^* \}$: In case that~$T^* \cup \{ \barcs_i^* \}$ contains a two-sided cycle,~$W^*_i$ equals this cycle. Otherwise, $T^* \cup \{ \barcs_i^* \}$ contains two one-sided cycles $C^*$ and a cycle containing $ \barcs_i^* $. In this case~$W^*_i$ walks once along~$C^*$, along a path in $T^*$ towards the one-sided cycle containing~$\barcs_i^*$, along this cycle, and finally back to~$C^*$ on the same path.
For the remainder of this section, we keep~$T, T^*, K, C, C^*, \barcs_1, \dots, \barcs_{g-1}$ and~$W^*_1, \dots, W^*_{g-1}$ fixed. 

Notice that we defined the above-described walks in such a way that each walk uses an edge at most twice. 
These~$g-1$ closed walks will yield the constraints needed to describe~$\R$-homology. 

\begin{lemma}
	\label{le:zHomologWeakProperty}
    Let~$z \in \R^A$ be a circulation in~$D$ and let~$\fav \in \R^F$ be an assignment on the~$\Pi$-facial walks such that~$z(a) = \partial \fav (a)$ holds for all arcs~$a \in T \cup \{ \barcs_1, \dots \barcs_{g-1} \}$. Then~$z$ is~$\R$-homologous to the all-zeros circulation.
\end{lemma}
\begin{proof}
    As~$z(a) = \partial \fav (a)$ holds for all arcs~$a \in T \cup \{ \barcs_1, \dots \barcs_{g-1} \}$,~$z - \partial \fav$ is a circulation in~$D$ that is zero on all arcs that are not contained in the spanning tree~$K$. Since~$K$ does not contain any cycle, the circulation~$z - \partial \fav$ must be zero on all arcs in~$K$ as well. Therefore,~$z = \partial \fav$, which yields the claim.
\end{proof}

In what follows, we will identify linear equations that ensure the existence of a vector~$\fav \in \R^F$ such that~$z(a) = \partial \fav (a)$ for all~$a \in T \cup \{ \barcs_1, \dots, \barcs_{g-1} \}$. To this end, the following properties of the matrix~$\partial$ will be useful.
Recall that the rows in~$\partial$ correspond to the arcs in~$A$ and the columns in~$\partial$ correspond to the walks in~$F$. We have~$\partial_{a,f} \neq 0$ if walk~$f$ uses the underlying undirected edge corresponding to~$a$ once or twice in the same direction. Since there are one-to-one correspondences between arcs~$A$ and dual edges~$A^*$ and between walks~$F$ and dual nodes~$V^*$, respectively,~$\partial$ may be interpreted as a matrix in~$\Z^{A^* \times V^*}$. With this interpretation we have~$\partial_{a^*,f^*} \neq 0$ if edge~$a^*$ is no loop and incident to~$f^*$ in the dual graph~$G^*$ or if $ a^* $ is a loop at $ f^* $ with negative dual signature.

\begin{lemma}
    \label{lem:determinantproperty}
    Let~$\partial_C$ be the submatrix of~$\partial$, consisting only of the rows and columns that correspond to nodes and edges used in~$C^*$, respectively.
    Then, we have~$| \det(\partial_C) | = 2$.
    Moreover, the absolute value of the determinant of any submatrix of~$\partial_C$, obtained by deleting exactly one row and one column, equals~$1$.
\end{lemma}
\begin{proof}
	In case that $ C^* $ is just a loop, consisting of the edge $ e^* $, the signature of $ e^* $ needs to be negative, since $ C^* $ is one-sided. Therefore the corresponding arc in $ D $ is used twice into the same direction in the corresponding $\Pi$-facial walk. Therefore $\partial_C = 2$.
	
	In case that $ C^* $ is not a loop, all entries in $\partial$ are either $ \pm 1 $ or zero. Moreover, for any arc $ a^* $ and node $ f^* $ in $ C^* $, the entry $\partial_{a^*,f^*} \neq 0$ if edge~$a^*$ is incident to~$f^*$ in the dual graph~$G^*$.
    Let~$C^* = (f_1^*, e_1^*, f_2^*, \dots, e_{\ell-1}^*, f_{\ell}^*, e_{\ell}^*, f_1^*)$ with~$\lambda^*(e_1^*) \cdot \cdots \cdot \lambda^*(e_{\ell}^*) = -1$. Notice that in each row and column in~$\partial_C$, there are exactly two non-zero coefficients, each of them equals either~$+1$ or~$-1$. 
    In a row that corresponds to an arc~$a$, the two non-zero coefficients have the same sign, if and only if, the walks in~$F$ use the underlying edge of~$a$ in the same direction. This is the case if and only if the dual signature of the dual edge of~$a$ is negative. Since~$C^*$ is one-sided, the number of rows in~$\partial_C$ in which the two non-zero entries have the same sign is odd. Notice that resorting rows/columns or multiplying rows/columns by~$-1$ does not change the absolute value of~$\partial_C$'s determinant.
    Moreover, multiplying rows/columns by~$-1$ does not change the parity of the number of rows in which the two non-zero entries have the same sign. Hence, the absolute value of~$\partial_C$'s determinant may be calculated as follows:
    \[
    | \det(\partial_C) | = \left\vert \det
    \begin{psmallmatrix}
    1 & 1 &   &        &        &   \\
    & 1 & 1 &        &        &   \\
    &   & 1 &    1   &      &   \\
    &   &   & \ddots & \ddots &   \\
    &   &   &        &    1   & 1 \\ 
    a &   &   &        &        & b             
    \end{psmallmatrix}
    \right\vert = 
    | b + (-1)^{\ell +1} a | = 2,
    \]
    where~$a$,~$b \in \{-1, +1\}$, the second equality follows from Laplace's formula and the last equality is due to~$a b = (-1)^{\ell +1}$. This holds true because~$C^*$ is one-sided, hence the displayed matrix should contain an odd number of rows, in which the two non-zero entries have the same sign.
    
    The second assertion follows by a similar argumentation, ending up with a triangular matrix. 
\end{proof}

\begin{lemma}
    \label{le:uniqueAlphaAlongOnetree}
    Given a circulation~$z \in \R^A$ in~$D$, there is a unique~$\fav \in \R^F$ such that~$z(a) = \partial \fav (a)$ holds for all arcs~$a \in T$.
\end{lemma}

\begin{proof}
    For any~$\fav \in \R^F$, let~$\fav_C$ denote the restriction of~$\fav$ to the~$\Pi$-facial walks that correspond to nodes of~$C^*$.
    Moreover, let~$z_C$ denote the restriction of~$z$ to the arcs in~$C$.
    By Lemma~\ref{lem:determinantproperty}, the determinant of~$\partial_C$ equals~$\pm 2$. Hence,~$\partial_C$ is regular and~$\fav_C$ is uniquely determined by the values of~$z_C$ via the equation~$z_C = \partial_C \fav_C$.
    The remaining values of~$\fav$ are uniquely determined by extending~$\fav_C$ along the arcs in~$T$ to~$\fav \in \R^F$ such that~$z(a) = \partial \fav (a) = s(a,h)\fav_h + s(a,g)\fav_g$ for every arc~$a \in T$ whose underlying edge is used in the two walks~$h$,~$g \in F$.
\end{proof}

We are now ready to characterize circulations that are~$\R$-homologous to~$\zerovec$.

\begin{lemma}
    \label{le:rhomologNon-OrientableSurface}
    A circulation~$z \in \R^A$ is~$\R$-homologous to~$\zerovec$, if and only if, the following~$g-1$ linear equations hold:
    \[
    \langle z, \xi(W^*_i) \rangle = 0 \qquad \textnormal{ for all } i \in [g-1].
    \]
\end{lemma}

\begin{proof}
    By Lemma~\ref{le:uniqueAlphaAlongOnetree}, let~$\fav \in \R^F$ be the unique assignment vector, which satisfies~$z(a) = \partial \fav (a)$, for all arcs~$a \in T$. Now, by Lemma~\ref{le:zHomologWeakProperty}, the circulation~$z$ is~$\R$-homologous to~$\zerovec$, if and only if,~$z(\barcs_i) = \partial \fav (\barcs_i)$, for all~$\barcs_i \in D \setminus (K \cup T)$. We show that for any~$i \in [g-1]$ the equation~$\langle z, \xi(W^*_i) \rangle = 0$ is equivalent to~$z(\barcs_i) = \partial \fav (\barcs_i)$. 
    
     We think of $\langle z, \xi(W^*_i) \rangle $ as walking along $ W^*_i $ and whenever crossing an arc $ a $ of $ D $ adding the value $ z(a) $ with the appropriate sign. 
     Since $ W^*_i $ is two-sided, the appropriate sign may be interpreted as an indicator for the direction $ z(a) $ crosses $ W^*_i $.
     The crucial property (besides $ W^*_i $ being two-sided) leading to the claimed equivalence is that all arcs corresponding to $ W^*_i $ except $ \barcs_i $ belong to $ T $. Therefore, $ z $ restricted to these arcs can be written as the combination of facial circulations with coefficients $ \fav $. Hence, at these positions, the values added when entering or leaving a face in $ D $ while walking along $ W^*_i $ will cancel out. 
     As a consequence, $\langle z, \xi(W^*_i) \rangle = 0$ if and only if the value added (with the appropriate sign, since $ W^*_i $ is two-sided) at $ \barcs_i $ also cancels out. This happens if and only if $z(\barcs_i) = \partial \fav (\barcs_i)$.
    
    Formally, consider the following two relationships:
    \begin{align}
    \label{eq:partialAlphaAndS}
    \partial \fav (a) &= s(a,f)\fav_f + s(a,g)\fav_g, \\
    \label{eq:sAndDualSignature}
    s(a,f) s(a,g)& = -\lambda^*(a^*),
    \end{align}
    where the underlying edge of~$a$ appears in the~$\Pi$-facial walks~$f$ and~$g$ and the dual edge of~$a$ is denoted by~$a^*$. Since a cyclic shift does not change our considerations, we assume~$W^*_i$ to be the following walk in the dual graph:~$(f_1^*, e_1^*, f_2^*, \dots, e_{\ell-1}^*, f_\ell^*, e_\ell^* = \barcs_i^*, f_1^*)$. 
    
    We denote the corresponding arcs in~$D$ by~$e_1, \dots, e_{\ell-1}, e_\ell = \barcs_i$.
    By the definition of~$\xi(\cdot)$, we have
    \begin{align*}
    \langle z, \xi(W^*_i) \rangle 
    &= \sum\limits_{a \in A} \left( \sum_{\substack{j \in \{1, \dots, \ell\} \\ a = e_j}} \lambda^*(e_1^*)\cdots\lambda^*(e_{j-1}^*) s(a,f_j) \right) \cdot z(a)  \\  
    &= \sum\limits_{j=1}^{\ell}\lambda^*(e_1^*)\cdots\lambda^*(e_{j-1}^*) s(e_j,f_j) \cdot z(e_j). \\
    \intertext{All arcs in~$W_i$ except~$\barcs_i$ lie in~$T$. For those arcs, we assume that~$z(e_j)=\partial \fav (e_j)$. Hence,}
    \langle z, \xi(W^*_i) \rangle & = \sum\limits_{j=1}^{\ell-1}\lambda^*(e_1^*)\cdots\lambda^*(e_{j-1}^*) s(e_j,f_j) \cdot \partial \fav (e_j)\\
    & \qquad \qquad + \lambda^*(e_1^*)\cdots\lambda^*(e_{\ell-1}^*) s(e_\ell,f_\ell) \cdot z(\barcs_i) \\
    &= \sum\limits_{j=1}^{\ell-1}\lambda^*(e_1^*)\cdots\lambda^*(e_{j-1}^*) s(e_j,f_j) \cdot \left( \fav_{f_j} s(e_j,f_j) + \fav_{f_{j+1}} s(e_j,f_{j+1}) \right) \\
    & \qquad \qquad + \lambda^*(e_1^*)\cdots\lambda^*(e_{\ell-1}^*) s(e_\ell,f_\ell) \cdot z(\barcs_i),\\
    \intertext{using Equation~\eqref{eq:partialAlphaAndS}. Now, by Equation~\eqref{eq:sAndDualSignature} we have} 
    \langle z, \xi(W^*_i) \rangle &= \sum\limits_{j=1}^{\ell-1}  \lambda^*(e_1^*)\cdots\lambda^*(e_{j-1}^*)  \cdot \left( \fav_{f_{j}} - \lambda^*(e_{j}^*) \fav_{f_{j+1}} \right) \\
    & \qquad \qquad + \lambda^*(e_1^*)\cdots\lambda^*(e_{\ell-1}^*) s(e_\ell,f_\ell) \cdot z(\barcs_i)\\[1ex]
    & = \fav_{f_1} - \lambda^*(e_1^*)\cdots\lambda^*(e_{\ell-1}^*) \fav_{f_{\ell}} \\[1ex]
    & \qquad \qquad + \lambda^*(e_1^*)\cdots\lambda^*(e_{\ell-1}^*) s(e_\ell,f_\ell) \cdot z(\barcs_i).
    \end{align*}
    Therefore,~$\langle z, \xi(W^*_i) \rangle=0$, if and only if,
    \begin{align*}
    z(\barcs_i) &= - \frac{\fav_{f_1} - \lambda^*(e_1^*)\cdots\lambda^*(e_{\ell-1}^*) \fav_{f_{\ell}}}{\lambda^*(e_1^*)\cdots\lambda^*(e_{\ell-1}^*) s(e_\ell,f_\ell)} \\
    &= - \left(\fav_{f_1} - \lambda^*(e_1^*)\cdots\lambda^*(e_{\ell-1}^*) \fav_{f_{\ell}}\right) \cdot \left(\lambda^*(e_1^*)\cdots\lambda^*(e_{\ell-1}^*) s(e_\ell,f_\ell)\right) \\
    \intertext{because the denominator is either~$+1$ or~$-1$. Hence,}
    z(\barcs_i)&= \fav_{f_{\ell}} s(e_\ell,f_\ell)  -  \lambda^*(e_1^*)\cdots\lambda^*(e_{\ell-1}^*) \fav_{f_1} s(e_\ell,f_\ell). \\
    \intertext{Using Equation~\eqref{eq:sAndDualSignature} and that $ W_i^* $ is two-sided, we obtain,}
    z(\barcs_i)&= \fav_{f_{\ell}} s(e_\ell,f_\ell)  + \lambda^*(e_1^*)\cdots\lambda^*(e_{\ell}^*) \fav_{f_1} s(e_\ell,f_1)  \\
    &= \fav_{f_{\ell}} s(e_\ell,f_\ell)  + \fav_{f_1} s(e_\ell,f_1)  \\
    &= \partial \fav (\barcs_i),
\end{align*}
which concludes the proof.
\end{proof}

Notice that even if the given circulation~$z$ in~$D$ is integer, the vector~$\fav \in \R^F$ defined in Lemma~\ref{le:uniqueAlphaAlongOnetree} is not necessarily integer.
The following lemma yields a characterization when~$\fav$ can be chosen to be integer.

\begin{lemma}
    \label{le:Zhomolog}
    Given a circulation~$z \in \Z^A$ in~$D$ and~$\fav \in \R^F$ such that~$z(a) = \partial \fav (a)$ for all arcs~$a \in T$,~$\fav$ is integral, if and only if,~$\sum_{a \in C} z(a) \equiv 0 \pmod* 2$.
\end{lemma}

\begin{proof}
    Observe that integrality extends along the arcs in~$T$ via the relation~$z(a) = \partial \fav (a) = s(a,h)\fav_h + s(a,g)\fav_g$ for every arc~$a \in T$ whose underlying edge is used in the two walks~$h$,~$g \in F$.
    Exploiting this fact, we fix one~$\Pi$-facial walk~$f \in F$ that corresponds to a node in the dual cycle~$C^*$.
    Now, it is sufficient to prove that~$\fav_f \in \Z$ is equivalent to~$\sum_{a \in C} z(a)$ is even. 
    Using the same notation as in the proof of Lemma~\ref{le:uniqueAlphaAlongOnetree}, value~$\fav_f$ is uniquely defined by~$z_C = \partial_C \fav_C$. By Cramer's rule,~$\fav_f = \sfrac{\det(\partial_{C_f})}{\det(\partial_C)}$, where~$\partial_{C_f}$ denotes the matrix obtained by replacing the column in~$\partial_C$ corresponding to~$f$ by~$z_C$.
    By Lemma~\ref{lem:determinantproperty},~$|\det(\partial_C) | = 2 $ and deleting one row and one column in~$\partial_C$ results in a matrix that has determinant~$\pm 1$.
    Therefore, by Laplace's rule applied to the column that equals~$z_C$, we conclude that~$\det(\partial_{C_f}) = \sum_{a \in C} \pm 1 \cdot z(a)$. 
    Since the signs in the summation do not affect the parity,~$\det(\partial_{C_f})$ is even if and only if~$\sum_{a \in C} z(a)$ is even. It follows that~$\sum_{a \in C} z(a)$ is even if and only if~$\fav_f$ is integer. 
\end{proof}

Lemma~\ref{le:uniqueAlphaAlongOnetree}, Lemma~\ref{le:rhomologNon-OrientableSurface}, and Lemma~\ref{le:Zhomolog} now yield the following characterization of~$\Z$-homology.

\begin{proposition}
	\label{prop:zhomologNonOrientableSurface}
	Let~$D = (V,A)$ be a digraph cellularly embedded in an orientable surface of Euler genus~$g$ and~$y$ an integer circulation in~$D$. 
	We can efficiently compute a one-sided cycle~$C$ and two-sided closed walks~$W^*_1, \dots, W^*_{g-1}$ in the dual graph of~$D$ (which do not use an edge more than twice) such that the following holds: 
	An integer circulation~$x \in \Z^A$ is~$\Z$-homologous to~$y$, if and only if,
	\[
	\langle x, \xi(W^*_i) \rangle = \langle y, \xi(W^*_i) \rangle \qquad \textnormal{ for all } i \in [g-1].
	\]
	and
	\[
	\sum_{a \in C} x(a) \equiv \sum_{a \in C} y(a) \pmod* 2
	\]
\end{proposition}

It remains to observe that Theorem~\ref{thm:zHomologousNonOrientable} is a direct consequence of this Proposition.

\section{Open questions} \label{sec:openquestions}

In Problem~\ref{probMain} we assume that the given cost is non-negative and require that the minimum-cost circulation~$x$ is non-negative (in addition to being integer and~$\Z$-homologous to~$y$).
One natural generalization of this problem arises by allowing arbitrary cost and/or imposing arbitrary bounds on~$x$.
For orientable surfaces, adapting the formulation in~\eqref{eq:mainProblemIP} results again in a totally unimodular system, and the discussion in Section~\ref{secOrientable} directly extends to this more general setting.
However, our algorithm for non-orientable surfaces in Section~\ref{sec:fixedgenus_polytimealg} does neither cover arbitrary cost nor general bounds on~$x$.
In fact, we do not know whether the generalized problem is polynomially solvable on fixed non-orientable surfaces.
In particular, we do not know the complexity status of deciding feasibility in this case.

Another natural generalization of this problem arises by considering $b$-flows instead of just circulations.
Again, for orientable surfaces a formulation according to~\eqref{eq:mainProblemIP} results in a totally unimodular system.
Even for fixed non-orientable surfaces, we believe that our algorithm in Section~\ref{sec:fixedgenus_polytimealg} can be adopted to the setting of $b$-flows with a polynomial running time in the case of a constant number of nodes with non-zero demand or supply.
However, we are not aware of a polynomial-time algorithm for arbitrary $b$.

For a matrix~$A$, let~$\Delta(A)$ denote the largest absolute value of the determinant of any square submatrix of~$A$.
We remark that the generalized problem on a fixed surface of Euler genus~$g$ can be recast as an integer program of the form~$\min \{c^\intercal x : Ax \le b, \, x \geq 0, x \in \Z^n\}$, where~$A$ and~$b$ are integer with~$\Delta(A) = 2^g$.
Such a formulation is given in~\eqref{eq:mainProblemIP}.
Using the fact that the Euler genus of a non-orientable surface is equal to the largest number of pairwise disjoint one-sided simple closed curves and Lemma~\ref{lem:determinantproperty}, one can show that~$\Delta(\partial) = 2^g$.
Determining the complexity status of integer programs with~$\Delta(A) \le \mathrm{const}$ is an open problem, see, e.g.,~\cite{ArtmannWZ2017}.
Artmann, Weismantel, and Zenklusen~\cite{ArtmannWZ2017} obtained a polynomial-time algorithm for solving integer programs with~$\Delta(A) \le 2$.
Therefore, their findings imply a polynomial-time algorithm for the generalized problem on the projective plane.
However, we are not aware of an efficient algorithm for the case of the Klein bottle, which has Euler genus~$g=2$.

Finally, we remark that while the running time of our algorithm is polynomial for fixed Euler genus $g$, the degree of the polynomial depends on $g$ (see, e.g., the cardinality of $\Omega$ in Section~\ref{sec:fixedgenus_polytimealg}).
We do not know whether Problem~\ref{probMain} is fixed-parameter tractable in $g$.

\subsubsection*{Acknowledgments} We would like to thank Ulrich Bauer and Joseph Paat for valuable discussions on this topic, and the five anonymous reviewers for their careful inspection of our paper and their detailed comments.

\bibliographystyle{plain}
\bibliography{references}

\end{document}